\documentclass[dvipsmx,10pt,draftcls,letter,onecolumn]{IEEEtran}
\usepackage{graphicx}
\usepackage{amssymb}
\usepackage{color}
\usepackage{amsxtra}
\usepackage{amsmath}
\usepackage{bm} 
\usepackage{cite} 
\usepackage{latexsym}
\usepackage{graphics}
\usepackage{verbatim}
\usepackage{epsfig}
\usepackage{subfigure}
\usepackage{setspace}
\usepackage{theorem}
\usepackage{algorithm}
\usepackage{algorithmic}
\newtheorem{theorem}{Theorem}[section]
\newtheorem{lemma}[theorem]{Lemma}
\newtheorem{proposition}[theorem]{Proposition}
\newtheorem{corollary}[theorem]{Corollary}

\newtheorem{definition}{Definition}[section]
\newtheorem{example}{Example}
\newtheorem{remark}{Remark}

\theorembodyfont{\normalfont\slshape}

\def\F{{\mathbb F}}

\def\R{{\mathbb R}}

\def\cC{{\mathcal C}}
\def\cD{{\mathcal D}}

\def\cL{{\mathcal L}}

\def\cP{{\mathcal P}}

\begin{document}

\thispagestyle{empty}

\title{A Graph Theoretical Analysis of Low-Power Coding Schemes for One-Hop Networks}

\author{Eimear~Byrne,
       and~Akiko~Manada,~\IEEEmembership{Member,~IEEE}

\thanks{E. Byrne is with School of Mathematics and Statistics, 
University College Dublin, Belfield, Dublin 4, Dublin, IRELAND,
email: ebyrne@ucd.ie.

A. Manada is with Graduate School of Informatics and Engineering, The Univ of Electro-Communications, Chofu, Tokyo, 182-8585 JAPAN, e-mail: amanada@uec.ac.jp}% <-this % stops a space
\thanks{This work was supported in part by Science Foundation Ireland (SFI) grant numbers 06/MI/006 and JSPS KAKENHI Grant Number 25870228.}%<-this % stops a space
\thanks{This paper was presented in part at 2011 IEEE International Symposium on Information Theory  (ISIT) \cite{BMMP:11}, and 2012 International Symposium on Communications and Information Technologies (ISCIT) \cite{BM:12}.}
\thanks{Manuscript received April 19, 2016}}

\markboth{IEEE Transactions on Information Theory}
{Byrne and Manada \MakeLowercase{}:A Graph Theoretical Analysis of Low-Power Coding Schemes for One-Hop Networks}

\maketitle

\begin{abstract}  
Coding schemes with extremely low computational complexity 
are required for particular applications, such as  
wireless body area networks, in which case both very high data accuracy and very low power-consumption are required features.  
In this paper, coding schemes arising from incidence matrices of graphs are proposed. An analysis of the resilience of
such codes to erasures is given using graph theoretical arguments; {\em decodability} of a graph
is characterized in terms of the rank of its incidence matrix. Bounds are given on the number of decodable subgraphs of a graph and the number of 
edges that must be deleted in order to arrive at an undecodable subgraph.  
Algorithms to construct codes that are optimal with respect to these bounds are presented.
\end{abstract}


%%%%%%%%%%%%%%%%%%%%%%%%%%%%%%%%%%%%%%%%%%%%%%%%%%%%%%%%%%%%%%

\section{Introduction\label{intro}} 
\IEEEPARstart{N}{etwork} coding \cite{ACLY:00, LYC:03} is by now a well established area offering demonstrated advantages
over routing in communications networks such as increased throughput and reduced latency.
Power consumption is a major concern for a number of applications.
For example, \emph{wireless body area networks} (WBANs) \cite{CGVCL:10, LYZ:10}, which are designed
to give reliable unobtrusive support for the monitoring of person's physiological data, 
favour ultra low-power coding schemes for communication between miniature sensors. 
At the same time, the number of re-transmissions requested due to errors should be 
minimized, since these are very costly from a power and/or memory perspective (see \cite{MAL:14} and the references therein).
Therefore, a major consideration in the design of any low-power coding scheme is that it be robust to packet loss, 
without incurring a large computational cost.

By a low-power coding scheme we mean one whose encoding functions have few arguments, and so in the linear case
can be represented by a sparse matrix. 
In this paper, we discuss measures of robustness of 
coding schemes when only the summation of two packets is allowed. 
The analysis for the case that packets are vectors over 
$GF(2)$ has been considered in \cite{BMMP:11, MP:09}. These schemes have the practical advantage of very low complexity encoding and decoding \cite{YMM:13}. Moreover, they can be identified with graphs and studied using that theory. In this paper, we extend these results to the more general case of coding schemes over arbitrary finite fields. The decoding criterion for a coding scheme over $GF(q)$ with $q$ odd is more general than for one over $GF(q)$ with $q$ even.

The rest of the paper is organized as follows. We present preliminaries in 
Section~\ref{background}, and identify a graph as a code via its incidence matrix. 
In Section~\ref{decodable}, we state a necessary and sufficient condition for a given coding scheme to 
be able to deliver all data at the terminal (in spite of packet erasures)
in terms of its corresponding graph representation, in which case we call the graph decodable.
In Section~\ref{decodable_probability}, we give a definition of {\em decoding probability} $P_G(y,z)$ for a given coding scheme as the evaluation of transmission success at 
$y=p,z=1-p$, where $p$ is the reliability of a packet. 
In particular, we derive a recurrence relation on $P_G(y,z)$ for certain coding schemes, which is similar to that for the well-known chromatic polynomial in graph theory. 
We then discuss in Section~\ref{dec_cut} the minimal size $b_G$ of a \emph{decoding cut}
in order to find a coding scheme with high decoding probability, 
and in Section~\ref{undec_graph} lower bounds on the number of undecodable graphs. 
Section~\ref{candidate} proposes some algorithms to produce 
suitable coding schemes, and Section~\ref{comparison} shows comparisons between various coding schemes.
Final remarks can be read in Section~\ref{conclusion}.

\section{Preliminaries\label{background}}
%\subsection{Preliminaries}
For the framework we consider, a system consists of the sender $S$, a terminal $T$ and a set of relays.  
The sender $S$ wishes to send $n$ data packets $p_1, p_2, \ldots, p_n $ in $GF(q)^{\ell}$ to the terminal $T$, via the relays. 
Throughout the paper, we always assume that $n\ge 2$. 
%where each $p_i$ is a vector %over the finite field $GF(q)$; 
%that is, $p_i$ is an element of 
%in $GF(q)^{\ell}$. 
%
\begin{comment}
\begin{figure}[htb]
\begin{center}
\includegraphics[width=60mm]{system_topology.pdf}
\caption
{A topology of a system consisting of senders, relays and the terminal }
\end{center}
\label{topology_figure}
\end{figure}
%\mbox{}\\[-24pt]
\end{comment}
%
At the relays, a total of $m$ encodings 
$$f_{1}(p_1, p_2, \ldots,p_n), f_{2}(p_1, p_2, \ldots,p_n), \ldots, f_{m}(p_1, p_2, \ldots,p_n)$$ are computed and transmitted to $T$.
%and sends $m(\ge n)$ packets in total to $T$.
The \emph{redundancy} $r=m/n$ is the ratio of the number of packets sent to $T$ to the number of original packets. 
A (linear) {\em  coding scheme} $\cC$ for the system is a collection of $GF(q)$-linear vectorial functions corresponding to packet encodings at relays; that is, $\cC=\{f_1, f_2, \ldots, f_m\}$. 
%In keeping with the desire to keep computational complexity to a minimum, 
We assume for the remainder that
%\begin{itemize} 
%\item[] 
each encoded packet $f_{i}(p_1, p_2, \ldots,p_n)$, $1\le i \le m$,  has the form $p_j$ or $p_j+p_{k}$ with $j\not=k$. 
%\end{itemize}
This scheme offers both low encoding and decoding complexity.

We next present some preliminaries on graphs (see \cite{W:96,D:10} for further reading).
Let $G=(V,E)$ be a finite (multi)-graph with vertex set $V=V(G)$ and edge multi-set $E=E(G)\subset V\times V$. 
An edge of $G$, $e=(x,y) \in V \times V$ is said to have initial vertex $x$ and end vertex $y$. Multi-edges may be distinguished by an edge labelling of $G$. 

A \emph{walk} starting at vertex $x$ and terminating at a vertex $y$ (or an $xy$-walk) is a sequence 
$\pi=x, e_{1},v_{1},e_{2},v_{2},\ldots,e_{\ell}, y$ such that each $e_{k}$ is an edge of $G$ with $e_{1}$ originating at $x$, $e_{\ell}$ ending at $y$ 
and such that the initial vertex $v_{k}$ of $e_{k+1}$ is the end vertex $v_{k}$ of $e_{k}$ for $k=1,...,\ell-1.$
If %the edges of the walk $\pi$ are distinct it is called a trail or simple walk and if 
the vertices of walk $\pi$ are all distinct it is called a \emph{path}. 
The number $\ell$ of edges appearing in a walk is referred to as its \emph{length}.
We write $x,v_{1},\ldots,v_{\ell-1}, y$ to denote any $xy$-path of the form $x, e_{1},v_{1},e_{2},v_{2}, \ldots,e_{\ell}, y$. If $G$ is simple (has no multiple edges)  
then $\pi$ is uniquely determined by the sequence of vertices $x,v_{1}, \ldots ,v_{\ell-1}, y$. A walk or path is called a \emph{closed walk} or \emph{cycle}, respectively, if its initial vertex is the same as its terminal vertex.
It is easy to see that an $xy$-path exists in $G$ if an $xy$-walk exists and that a closed walk of odd length implies the existence of a cycle of odd length. 
%$\bm \pi: v_i=w_0, w_1, \ldots, w_l=v_j$  (for some integer $l\ge 1$) 
%such that $(w_i, w_{i+1}) \in E$ for each $0\le i\le l-1$. 
We say $x$ is \emph{connected to} $y$ in $G$ when a walk from $x$ to $y$ exists in $G$.  Graph $G$ is called \emph{connected} if each pair of vertices are connected, otherwise $G$ is called \emph{disconnected}. A \emph{connected component} (or simply \emph{component}) of $G$ is a maximal connected subgraph of $G$, 
where a subgraph $H$ of $G$ is a graph such that $V(H) \subset V(G)$ and $E(H)\subset E(G)$ hold.

%In particular, a path is called a \emph{cycle} when $v_i$ (the starting vertex $w_0$ of a path) is equal to $v_j$ (the terminal vertex $w_l$ of a path).  

The \emph{adjacency matrix} $A=A_G$ of $G$ is the
$|V|\times |V|$ integer matrix whose $(i,j)$-entry 
is the number of edges with initial vertex $x$ and end vertex $y$. 
We write $\Omega(G)$ to denote the maximum multiplicity of any edge joining a pair of vertices of $G$; that 
is, the maximum entry in $A$. The \emph{incidence matrix} $B=B_G$ of $G$ is the $|E|\times |V|$ 0-1 matrix whose $(i,j)$-entry is 1 if edge $e_i$ is incident with vertex $v_j$
with respect to some labelling of the vertices and edges of $G$. 
Observe that the $(i,j)$-entry of $A^\ell$ corresponds to the 
number of walks of length $\ell$ from $v_i$ to $v_j$.

%The \emph{incidence matrix} $B_G$ of $G$ is 
%the $|E|\times |V|$ integer matrix, where $i$-th row and $j$-th column correspond to an edge $e_i$ and a vertex $v_j$ in $G$, respectively, 
%and $(i, j)=1$ if and only if $e_i$ is incident to  $v_i$  

%A \emph{path} of length $\ell$ between a pair of vertices $v,v' \in V$ is a sequence $v=v_1,...,v_{\ell+1}=v'$ for distinct edges
%$v_iv_{i+1} \in E$ for $i=1,...,\ell$. If $v=v'$ the path is called a \emph{cycle}.

We write $d_L(v)$ to denote the number of loops (length-1 cycles) incident with a given vertex $v$, and we let 
$\Delta_L(G):=\max\{ d_L(v) : v \in V\}$.
We denote by $L_G=\sum_{v\in V}d_L(v)$  the number of loops  of 
$G$. 
We define the \emph{incidence degree} of a vertex $v$, expressed $d_{I}(v)$, as  
the number of edges incident with $v$. Each loop at $v$ contributes a count of one to $d_I(v)$.
We define $\delta _{I}(G):=\min \{d_{I}(v): v \in V\}$. 
The sum of the incidence degrees $S_I(G)$ satisfies $S_I(G):=\sum_{v \in V}d_I(v) = 2|E|-L_G$.

Given a connected graph $G$, the \emph{edge-connectivity} $\lambda(G)$ of $G$ is the smallest number of edges such that the resulting graph formed by deleting those edges is disconnected. 
Observe that $\lambda(G) \leq \delta _I(G)$ since deleting all edges attached to a vertex $v$ with incidence degree $d_{I}(v)= \delta _{I}(G)$
makes $v$ isolated. 
%A \emph{subgraph} $H$ of $G$ is a graph such that $V(H) \subset V$ and 
%$E(H) \subset E$. A subgraph $H$ is called the subgraph of $G$ \emph{induced by} $S\subset V$ if $V(H)=S$ if any edge in $G$ whose
%endpoints are both in $V(H)$ is an edge in $H$.

A graph $G=(V,E)$ is called a \emph{bipartite graph} or \emph{bi-colourable} if $V$ can be partitioned into two sets $U$ and $W=V \backslash U$ such that
every edge of $G$ is incident with exactly one member of $U$ and one member of $W$; that is, no two vertices in the same set 
are connected by an edge. (Note that a graph with loops is not bipartite, nor indeed colourable.)
There are several characterizations of bipartite graphs. One that is relevant to this paper is K\"{o}nig's well-known result (see \cite[Chapter 1]{W:96}).
%mbox{}\\[-18pt]

\begin{theorem}\label{thbip}
    A graph is bipartite if and only if it has no cycle of odd length. 
\end{theorem}
%\mbox{}\\[-20pt]

We now relate identify a coding scheme with a multi-graph. 
Consider the case that the sender $S$ sends $n$ packets $p_1, p_2, \ldots, p_n\in GF(q)^{\ell}$ to the terminal $T$ with redundancy $r$. 
Given a coding scheme
 $\cC=\{f_1,f_2, \ldots,f_{m}\}$
we generate a \emph{multi-graph representation} $G=G_{\cC}$ for $\cC$ as follows.
\begin{enumerate} 
\item $G$ has as vertices $p_1, p_2, \ldots, p_n$.
\item For $j\not =j'$, $(p_j,p_{t})$ is an edge of $G$ if $f_i(p_1, p_2, \ldots,p_n)=p_j+ p_{t}$ for some $f_i$ (\emph{i.e.}, if $p_j+ p_{t}$ is sent to the terminal $T$).
\item $G$ has a loop at $p_j$ if $f_i(p_1,p_2, \ldots,p_n)=p_j$ for some $f_i \in \cC$ (\emph{i.e.}, if $p_j$ itself is sent to $T$).
\end{enumerate}

An example of a coding scheme $\cC$ and its corresponding graph 
representation is given in Table~\ref{interencoding_table} and Figure~\ref{sys_fig}. 
%the citation of the figure should be modified.
The erasure of packets during a transmission is identified with deletions of corresponding edges in $G$. 
Also, for a case with $n$ packets and redundancy $r$, any graph representation of a corresponding coding scheme must have $n$ vertices and $m=rn$ edges.

\begin{table}[htb]
\begin{center}
\scriptsize
\begin{tabular}{|c|c|c|c|c|c|}
\hline
  \multicolumn{6}{|c|}{Encoding at the sender $S$} \\
\hline
$p_1$    & $p_2$ & $p_3 +  p_4 $ &  $p_4 +p_5  $ &  $p_5 +p_6 $ &  $p_6 +  p_1$ \\
\hline
$p_4$    & $p_5$ & $p_6 +  p_7$ &  $p_7 + p_8$ &  $p_8 + p_9$ &  $p_9 + p_4$ \\
\hline
$p_7$    & $p_8$ & $p_9 + p_{10}$ &  $p_{10}+ p_{11}$ & $p_{11} +  p_{12}$ & $p_{12} + p_{7}$\\
\hline
$p_{10}$ & $p_{11}$ & $p_{12}+ p_{1} $ &  $p_1 +p_{2}$  & $p_2 +p_{3} $ & $p_3 +p_{10}$ \\
\hline
\end{tabular}
\end{center}
\caption{An example of a coding scheme $\cC$ \label{interencoding_table}}
\end{table}

%%%%%%

\begin{figure}[htb]
\begin{center}
\includegraphics[width=65mm]{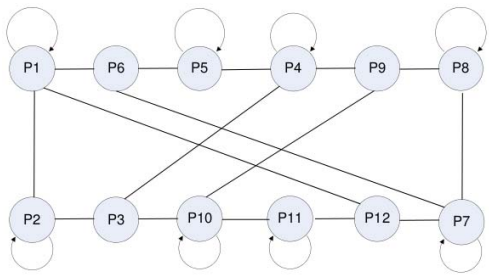}
\caption
{The graph representation $G=G_{\cC}$ of coding scheme $\cC$ in Table~\ref{interencoding_table} \label{sys_fig}}
\end{center}
\end{figure}
\mbox{}\\[-25pt]

\section{Decodable Graphs\label{decodable}}
We provide a necessary and sufficient condition on the graph representation of a coding scheme for full packet retrieval at the terminal $T$
extending \cite[Theorem III.1]{BMMP:11}. 
For the remainder of this paper, when a case with $n$ packets and redundancy $r$ is given,
we let $\cC=\{f_1,f_2,\ldots, f_{rn}\}$ denote a coding scheme over $GF(q)$ for the system, 
where each $f_i$ is an encoding of packets $p_1,p_2, \ldots,p_n\in GF(q)^\ell$, and we let $G=G_{\cC}$ denote the graph representation of $\cC$.
$H$ will denote a subgraph of $G$ formed by deleting a corresponding edge of $G$ for every erasure occurring in a given transmission.
%\mbox{}\\[-18pt]

\begin{definition}
    $H$ is called decodable if the linear system of equations corresponding to the edges of $H$ can be solved for unique $p_1,p_2, \ldots,p_n$.
    Otherwise $H$ is called undecodable.
\end{definition}
%\mbox{}\\[-20pt]

Throughout this paper, we denote by $\cD(n,m)$ the set of decodable graphs on $n$ vertices and $m$ edges. We always set $m\geq n$ since it is impossible to uniquely retrieve all $n$ packets from 
$m$ received packets if $m<n$. One criterion for decodability is given by the following. \\[-18pt]

\begin{lemma}\label{lemconn} 
    Suppose that the terminal $T$  has retrieved the packet $p_i$ for some $i \in \{1,2, \ldots,n\}$. 
    Then for every $p_j$ connected to $p_i$ in $H$, the packet $p_j$
    can be retrieved by $T$.  
\end{lemma}
%\mbox{}\\[-4ex]
\begin{proof}
Let $p_j$ be connected to $p_i$ in $H$, say by a path $\pi: p_i=p_{i_1},p_{i_2},\ldots,p_{i_h}=p_j$.
    This corresponds to the linear system of $h$ equations 
    $g_1=p_{i_1},g_2=p_{i_1}+p_{i_2},\ldots,g_{h}=p_{i_{h-1}}+p_{i_h},$ which has rank $h$
    and can be solved at $T$.    
\end{proof}%\mbox{}\\[-20pt]

\begin{lemma}\label{lemdh}
   $H$ is decodable if and only if its incidence matrix $B_H$ has rank $n$, 
   considered as a matrix over $GF(q)$. 
\end{lemma}
%\mbox{}\\[-4ex]
\begin{proof}
    Let $P$ be the $n \times \ell$ matrix whose rows are the vectors $p_1,p_2, \ldots,p_n$.
    The encoded packets received by the terminal decoder are the rows of the $|E(H)|\times \ell$ matrix $M=B_H P$.
    The equation $M=B_H P$ can be solved for unique $P$ if and only if $B_H$ has rank $n$ over $GF(q)$. 
\end{proof}%\mbox{}\\[-20pt]

A criterion to determine whether or not the incidence matrix of a graph has full rank as a real matrix is already known. For example, in \cite{N:76} it was shown that the incidence matrix of a loop-free graph has full $\R$-rank if and only if it has no bipartite connected components nor isolated vertices. It is not hard to see that, for odd $q$, the argument given there holds also for the $GF(q)$-rank of an incidence matrix (cf. \cite{N:76}), with a minor modification to include the case of a graph with loops.\\[-18pt]

\begin{theorem}\label{thimrk}
   Let $\F$ be a field of characteristic not equal to 2. The rank over $\F$ of 
   the incidence matrix of an undirected graph with $n$ vertices,
   $s$ bipartite components and $t$ isolated vertices is $n-s-t$.  
\end{theorem}%\mbox{}\\[-18pt]

Applying Lemma~\ref{lemdh} and Theorem~\ref{thimrk}, for fields of odd characteristic we have the following characterization of decodable subgraphs.%\mbox{}\\[-18pt]

\begin{corollary}\label{cordec}
Let $\cC$ be a coding scheme over $GF(q)$ for odd $q$. Then the subgraph $H$ of $G$ is decodable if and only if each connected component of $H$ is neither  a bipartite graph nor an isolated vertex.
\end{corollary}%\mbox{}\\[-20pt]

To summarize, we obtain the following theorem. \\[-18pt]

\begin{theorem}  
Let $\cC$ be a coding scheme over $GF(q)$, and $G$ be its graph
representation. Then 
a subgraph $H$ of $G$ is decodable if and only if exactly one of the following holds:
\begin{enumerate}
   \item
      $q$ is even and every connected component of $H$ has a loop,
   \item
      $q$ is odd and every connected component of $H$ has an odd cycle.    
\end{enumerate}
\label{q2_thm}
\end{theorem}
%\mbox{}\\[-2.5ex]
\begin{proof}
Let $U$ be a connected component of $H$.
If $U$ has a loop, then from Lemma \ref{lemconn} each packet ${p}_i \in V(U)$ can be retrieved.
 Suppose that $q$ is even and that $U$ has no loop. 
Then every
	row of $B_U$ has Hamming weight exactly $2$ so the sum of the columns of $B_U$ is zero over $GF(q)$ and the system has rank
	less than $|V(U)|$.
     	%the sum of the columns of $B_U$ is zero over $GF(2)$ if (and only if)
 %$U$ has no loops, and hence, $B_U$ is not full-rank. 
%Conversely, when $U$ has a loop, then columns in $B_U$ are linearly independent, so that $B_U$ is full-rank. 
%Indeed, if the sum of some columns in $B_U$ is zero, then any vertices corresponding to the columns have no loops and no edges connecting to the other vertices in $U$, which contradicts the existence of a loop or $U$ being connected.} 
   If $q$ is odd, then from Corollary \ref{cordec} and Theorem~\ref{thbip}, $U$ has an odd cycle if and only if 
it is decodable.
\end{proof} 
%\mbox{}\\[-20pt]

The following is now immediate from Theorem~\ref{q2_thm}.%\\[-18pt]

\begin{corollary} 
Let $\cC$ be a coding scheme over $GF(q)$ for odd $q$. 
Then a connected subgraph $H$ of $G$ is decodable if and only if the adjacency matrix $A=A_H$ of $H$ satisfies $\mathrm{Tr}A^\ell>0$
for some odd number $\ell \in \{1,2,\ldots,n\}.$
\label{trace_cor}
\end{corollary}
%\mbox{}\\[-4ex]
\begin{proof}
The $(i,j)$-entry of $A^\ell$ is the number of 
walks of length $\ell$ from $v_i$ to $v_j$. 
Hence, $\mathrm{Tr}A^\ell>0$ holds for some odd number $\ell$ 
if and only if there exists a closed walk of odd length $\ell$ (and hence, if and only if there exists an odd cycle). The result now follows from Theorem \ref{q2_thm}.
\end{proof}
%\mbox{}\\[-20pt]

The robustness of a coding scheme against packet loss can be measured as a function of 
the number of decodable subgraphs found upon deleting some edges.
We will next define and discuss such a function.

\section{The Decoding Probability\label{decodable_probability}}

We first recall definitions provided in \cite{BMMP:11}. Given a graph $G\in \cD(n,m)$, 
we denote by $c_x^G$ the number of decodable subgraphs of $G$ formed by deleting $x$ edges of $G$, and we write $u^G_x = \binom{m}{x}-c_x^G$ to denote the number of undecodable subgraphs of $G$
found by deleting some $x$-set of its edges. 
(Here, we assume that edges in $G$ are labeled distinctly in order to distinguish them. Thus, deleting different edges  produces 
different subgraphs at all times.)
By convention, we set $c_x^G=0$ if $x>m$ or $x<0$.
We define the \emph{decoding probability} of $G$ by $\cP_G(p,1-p)$,
where 
\begin{eqnarray}
\cP_G(y,z):=\sum_{x=0}^{m}c_x^G y^{m-x}z^x,
\label{dp_eqn}
\end{eqnarray}
and $p$ is the probability that an edge is not deleted; that is, the probability that 
a packet is successfully transmitted to the terminal.
%and $Y=1-X$ the probability that an edge is deleted 
%(\emph{i.e.}, the probability that a packet is erased during the transmission). 
In other words, $1-p$ is equal to the packet loss rate for each packet. 
It is clear that the decoding probability 
is the probability that all packets are retrieved at the terminal even though 
some packet loss occurs during transmission, and indeed,   
measures the robustness of a coding scheme against packet loss. 
%Indeed, our interest is to construct a coding scheme $\cC$ for a fixed system whose corresponding 
%graph $G=G_{\cC}$ has a high decoding probability. 

In this section, we will first make an approach for 
analyzing the decoding probability of a graph $G$,
using pre-computed decoding probabilities of graphs $G-e^j$ and $G\cdot e^k$
obtainable by the \emph{deletion} and the \emph{contraction} of edges between $u$ and $v$, respectively, as denoted below.

For a graph $G$ and distinct vertices $u,v$ in $G$, suppose that 
there are $k \ge 1$ edges $e_1,e_2, \ldots, e_k$ connecting $u$ and $v$.  
Set $e^j=e^j(u,v)=\{e_1,e_2, \ldots, e_j\}$ for $1\le j \le k$. 
We define $G-e^j$ $(1\le j \le k)$ to be the graph obtained by deleting $j$ edges from 
those $k$ edges, 
and $G\cdot e^k$ to be the graph obtained by identifying the ends $u,v$  as a single vertex (\emph{i.e.}, by applying the contraction of edges between $u$ and $v$).
The following lemma states a relationship between $G$, $G-e^j$ and $G\cdot e^k$.
%\\[-18pt]

\begin{lemma}
Let $G\in \cD(n,m)$ be a graph representation of a coding scheme over $GF(q)$, and 
suppose there are $k$ edges 
$e_1, e_2 \ldots, e_k$ connecting distinct vertices $u$ and $v$ in $G$. Then for the case when 
\begin{enumerate}
\item If $q$ is even; or
\item If $q$ is odd and no edge in $e^k(=e^k(u,v))$ can be included in a cycle in $G$,
\end{enumerate}
we have $$c_x^G=\sum_{j=1}^k \binom{k}{j}c_{x-j}^{G-e^j}+c_{x}^{G\cdot e^k}.$$
\label{equality_lem}
\end{lemma}
%\mbox{}\\[-5ex]%% margine arrangement!!
\begin{proof} 
Observe that there are $k+1$ mutually exclusive cases when deleting 
$x$ edges from $G$; 
\begin{itemize}
\item[(A)] the $x$ edges 
include $j$ edges $(1\le j \le k)$ from  $e^k$; or 
\item [(B)] the $x$ edges do not include any edge in $e^k$. 
\end{itemize}
It is trivial that 
the number of decodable graphs in case (A) is given by $ \binom{k}{j}c_{x-j}^{G-e^j}$, 
here $\binom{k}{j}$ comes from the number of choices deleting $j$ edges from $e^k$ 
and observe that deletion of any such $j$ edges results in an isomorphic graph (a graph with the same topology). 
So we will show that the number of decodable graphs counted in case (B) is 
equal to  $c_{x}^{G\cdot e^k}$. To do so, we define $X$ and $Y$ to be 
the sets of decodable graphs counted in case (B) and in $c_{x}^{G\cdot e^k}$, respectively, 
and then show the existence of a bijection $f: X \rightarrow Y$. 

Let $H$ be a subgraph of $G$ 
such that $e^k\subset  E(H)$. Then, it is straightforward to check that $H$ and $H\cdot e^k$ have 
the same number of components and loops at each component. 
Furthermore, for $GF(q)$ with $q$ odd, 
from the assumption on $e \in e^k$, we have that $H$ 
has an odd cycle if and only if $H \cdot e^k$ has an odd cycle. 
Hence, for any $q$, we can conclude that $H$ is decodable if and only if $H \cdot e^k$ is decodable. 
Therefore, set $f: X \rightarrow Y$ to be $f(H)=H \cdot e^k$, and then $f$ is obviously bijective, which completes the proof. 
\end{proof} %\mbox{}\\[-20pt]

From Lemma~\ref{equality_lem}, we can obtain 
the following corollary which indicates the relationship of 
 $\cP_G$, $\cP_{G-e^j}$ and $\cP_{G\cdot e^k}$.
 A similar result holds for the chromatic polynomial in graph theory.\\[-18pt]

\begin{corollary}
Let $G\in \cD(n,m)$ be a graph representation of a coding scheme over $GF(q)$, and 
suppose there are $k$ edges 
$e_1, e_2 \ldots, e_k$ connecting distinct vertices $u$ and $v$ in $G$. 
If assumption 1) or 2) in the statement of Lemma~\ref{equality_lem} holds, 
then we have 
$$\cP_G(y,z)=\sum_{j=1}^k \binom{k}{j}z^j\ \cP_{G-e^j}(y,z)+y^k\ \cP_{G\cdot e^k}(y,z).$$
\label{euation_cor}
\end{corollary}
%\mbox{}\\[-4ex]
\begin{proof} 
\begin{eqnarray}
\cP_G(y,z)    &=& \sum_{x=0}^{m} c_x^G y^{m-x}z^x \nonumber \\%= X^m+\sum_{x=1}^{m} c_x^G X^{m-x}Y^x\\
      &=& \sum_{x=0}^{m} \left(\sum_{j=1}^k \binom{k}{j} c_{x-j}^{G-e^j}+c_{x}^{G\cdot e^k}\right) y^{m-x}z^x \\
      &=& \sum_{j=1}^k  \binom{k}{j} \sum_{x=0}^{m} c_{x-j}^{G-e^j} y^{m-x}z^x+ 
               \sum_{x=0}^{m} c_{x}^{G\cdot e^k} y^{m-x}z^x \nonumber \\
      & = &  \sum_{j=1}^k  \binom{k}{j} z^j \sum_{x=0}^{m-j} c_{x}^{G-e} y^{m-x-j}z^{x}  \\
      & & \ \ \ \ \ \ +y^k\sum_{x=0}^{m-k} c_{x}^{G\cdot e^k} y^{m-x-k}z^x  \\
      &=& \sum_{j=1}^k \binom{k}{j}z^j\ \cP_{G-e^j}(y,z)+y^k\  \cP_{G\cdot e^k}(y,z), \nonumber
\end{eqnarray}
where (2) is from Lemma~\ref{equality_lem}, (3) is obtainable from 
\begin{eqnarray*}
\sum_{x=0}^{m} c_{x-j}^{G-e^j}y^{m-x}z^x &=& \sum_{x=-j}^{m-j} c_{x}^{G-e^j}y^{m-x-j}z^{x+j}\\
&=&z^j\sum_{x=0}^{m-j} c_{x}^{G-e^j}y^{m-x-j}z^x %(\mbox{as }  c_{x}^{G-e^j}=0 \mbox{ with } x<0)
\end{eqnarray*}
observing that  $c_{x}^{G-e^j}=0$ when $x<0$,
and (4) is from the fact that $c^{G\cdot e^k}_x=0$ for $x>m-k$ as 
$G\cdot e^k$ has $m-k$ edges. 
\end{proof} %\mbox{}\\[-20pt]
 
Corollary~\ref{euation_cor}  implies that for $G\in \cD(n,m)$ 
if we have computed decoding probabilities of 
graphs in $\cD(n',m')$ with $n'\le n$  and $m'<m$ in advance, 
the decoding probability of $G$ can be obtained recursively. 
Since $p>1-p$ in general, if $\cP_{G\cdot e^k}(p,1-p)$ is large enough, then  $\cP_{G}(p,1-p)$ could also have large value. 

%We cannot extend Lemma~\ref{equality_lem} and Corollary~\ref{euation_cor} for the case of odd $q$ since it might be the case that the contraction of an edge $e$ produces some odd cycles in a component which has originally no odd cycles or may vanish all odd cycles of a component.

\section{Decoding Cuts\label{dec_cut}}
Our interest is to find a graph $G\in\cD(n,m)$ (for some fixed $m$ and $n$) which has the maximum decoding probability amongst all graphs in $\cD(n,m)$.  
To this end, we define
a \emph{decoding cut} $\cL$
of $G$, as follows.
Given $G \in \cD(n,m)$,  
a decoding cut (D-cut) $\cL$ is a subset  of $E$
such that graph $\tilde G=(V,E\backslash \cL)$ is undecodable. 
We denote by $b_G$ the smallest cardinality of 
any D-cut of $G$. In other words, 
$b_G$ is the smallest $x$ such that 
$u_x^G>0$ (that is, $c_x^G<\binom{m}{x}$). Clearly  
a high value of $b_G$ is desirable for a good coding scheme.

\begin{comment}
%I do not know why I added the following before...
The following upper bound comes from Lemma~\ref{equality_lem}.
\begin{corollary}

\end{corollary}
\end{comment}

\subsection{Upper bounds on $b_G$ for $GF(q)$ with $q$ even}
\label{b_G-even}
We first discuss some upper bounds on $b_G$ for $GF(q)$ with $q$ even. Recall from Theorem~\ref{q2_thm}
that a graph $H$ is decodable  if and only if each component of $H$ has a loop. 

\begin{remark} For $b_G$ of a graph $G$, we note the following.
\begin{enumerate}
\item $b_G\leq \min(L_G, \delta _{I}(G))$ since 
deleting all loops in $G$ or deleting all edges attached to a vertex $v$ with incidence degree $d_{I}(v)= \delta _{I}(G)$
yields an undecodable graph. 
\item If $L_G\geq \lambda(G)$, then $\lambda(G) \leq b_G$
since a resulting graph $\tilde G$ of $G$ after deletion of some edges cannot be undecodable, if $L_{\tilde G}\not =0$, 
unless $\tilde G$ is disconnected. 
\end{enumerate}
\label{minloop_rem}
\end{remark}

The following lemma can be easily obtained
using elementary graph theory.

\begin{lemma}
Let $G\in \cD(n,m)$. Then, $\delta _{I}(G)\leq 2m/n-1$ and $b_G \leq 2m/(n+1)$.
In particular, $b_G \leq \min(  \lfloor 2m/n-1 \rfloor , \lfloor 2m/(n+1)  \rfloor)$
\label{edge_lem}
\end{lemma}
%\mbox{}\\[-4ex]
\begin{proof}
Since $G$ is decodable, $L_G \geq 1$, and we have
$ n \delta_I(G) \leq S_I(G) = 2m - L_G < 2m.$
Furthermore, since $b_G\leq L_G$, we have 
$n b_G \leq n \delta_I(G) \leq S_I(G) = 2m - L_G < 2m -b_G$. 
\end{proof} \mbox{} %\\[-20pt]

From Lemma~\ref{edge_lem}, it follows that for a system sending $n$ packets with redundancy $r$, any graph representation $G$ of a coding scheme  satisfies
$b_G\leq \min(2r-1, \lfloor \frac{2rn}{n+1}\rfloor) = \lfloor \frac{2rn}{n+1}\rfloor$, which is simply 
$2r-1$ whenever $r\leq \frac{n+1}{2}$. 
%As for $M:=\min(2r-1, \frac{2rn}{n+1})$, observe that $M=2r-1$ when $r<\frac{n+1}{2}$ 
%and in general, the inequality holds. So at this moment, we assume $M=2r-1$.

\subsection{Upper bounds on $b_G$ for $GF(q)$ with $q$ odd}
We next discuss some upper bounds on $b_G$ for $GF(q)$ with $q$ odd, where the notion of cuts in graph theory plays an important role.

Recall that a \emph{cut} of a graph $G=(V,E)$ is a partition of the vertex set $V$ into a pair of disjoint subsets $V_1$ and $V_2$. The \emph{cut-set} of the partition is the set of edges of $E$ that have one endpoint in $V_1$ and the other in $V_2$. A cut is called a \emph{maximal cut} if no other cut has a larger cut-set.
A maximal cut of a graph can be identified with a maximal bipartite subgraph as follows. A subgraph $(V,E')$ of $G$ is bipartite
if and only if its edge set $E'$ form a cut-set of a partition of the vertices of $V$. 
Then a bipartite subgraph has a maximum number of edges when its edges correspond to a maximal cut of $G$. We denote this number by $\Gamma(G)$.\\[-18pt]

\begin{lemma}\label{dcut_lem}
   The size of the minimum D-cut $b_G$ of $G=(V,E)$ is upper-bounded by the following:
   \begin{itemize}
      \item[(i)]
         $\delta_{I}(G)=\min \{d_{I}(v): v \in V\}$,
      \item[(ii)]
         $|E|-\Gamma(G)$.
      %\item[(iii)]
       %  $O_G$, where $O_G$ denotes the minimum number of edge-disjoint odd cycles in $G$,    
   \end{itemize} 
\end{lemma}
%\mbox{}\\[-4ex]
\begin{proof}
    The first item is clear since we can make an isolated vertex by deleting $\delta_I(G)$ attached edges from a vertex $v$ with incidence degree $d_I(v)=\delta_I(G)$, and the isolated vertex is unretrievable.
    
    To see that (ii) holds, note that to produce a bipartite subgraph
    $(V,E \backslash \cL)$ of $G$ by deleting edges $\cL$ from $E$ we must have
    $|\cL| \geq |E| - \Gamma(G)$ and we have equality exactly when $(V,E \backslash \cL)$ is a maximal bipartite subgraph of $G$.
\end{proof} %\mbox{}\\[-20pt]

%\newpage
There are several bounds on the size of a maximal bipartite subgraph. We will use the following results of Edwards (cf. \cite{E:73}, \cite{E:75} and \cite{KW:08}).\\[-18pt]

\begin{theorem}[Edwards]\label{mc_thm}
   Let $G=(V,E)$ have $n$ vertices and $m$ edges. Then 
   \begin{enumerate}
      \item
          $\Gamma(G) \geq \frac{m}{2} + \frac{1}{8}(\sqrt{8m+1}-1)$
      \item
          $\Gamma(G) \geq \frac{m}{2} + \frac{n-1}{4}$ if $G$ is connected.    
   \end{enumerate} 
\end{theorem}\mbox{}\\[-20pt]

Applying Theorem \ref{mc_thm} to (ii) of Lemma \ref{dcut_lem} immediately gives the following bounds on $b_G$.
\\[-18pt]

\begin{corollary}
   Let $G=(V,E)$ have $n$ vertices and $m$ edges. Then
   \begin{enumerate}
      \item
          $b_G \leq \frac{m}{2} - \frac{1}{8}(\sqrt{8m+1}-1)$, and 
      \item
          $b_G \leq \frac{m}{2} - \frac{n-1}{4}$ if $G$ is connected.    
   \end{enumerate} 
\end{corollary}\mbox{}%\\[-20pt]

The following lemma gives us an 
upper bound on $b_G$ from the perspective of the 
minimum incidence degree $\delta_I(G)$.%\\[-18pt]

\begin{lemma}
For $G\in \cD(n,m)$, $\delta _{I}(G)\leq 2m/n$, and 
therefore, $b_G \leq 2m/n$.
\label{edgenumber_lem}
\end{lemma}
%\mbox{}\\[-4ex]
\begin{proof} 
Recall that the sum of 
the incidence degrees $S_I(G)$ in $G$ is given by $S_I(G)=2|E(G)|-L_G\leq 2|E(G)|$. Therefore, $ n \delta_I(G) \leq S_I(G) \leq 2m$, which 
implies $\delta_I(G)\leq 2m/n$. 
\end{proof}%\mbox{}%\\[-20pt]

\section{Lower Bounds on the Number of \\ Undecodable Graphs\label{undec_graph}}
We have discussed upper bounds of $b_G$, which counts the number of edges that must be deleted to produce a undecodable subgraph. 
In this section, we give lower bounds on the number of undecodable subgraphs $u_x^G$ of $G$ (or equivalently, upper bounds on the number of decodable subgraphs $c_x^G$ of $G$) for $x\ge b_G$. 
%In this section, we will provide some lower bounds on $u_x^G$ when $x$ is close to $b_G$.

\subsection{For $GF(q)$ with $q$ even} 
We first present 
the following lemma which describes a sharp lower bound on $u_{b_G}^G$.

\begin{lemma}\label{lemub1}
Let $G$ be a decodable graph with $n$ vertices and $m$ edges.
Then for any $G \in {\cD(n,m)}$, we have 
$$u^G_{b_G} \geq b_G(n+1)+n-2m+1.$$ 
%$c^G_{b_G}\leq \binom{m}{b_G} - b_G(n+1)-n+2m-1$. 
\end{lemma}
%\mbox{}\\[-4ex]
\begin{proof} 
First recall $b_G \leq \min(L_G,\delta_I(G))$.
Let $\alpha $ be the number of vertices with incidence degree $b_G$ and let
$\beta$ be the number of vertices with incidence degree at least $b_G+2$. 
Then by considering the sum of incidence degrees, we have 
$$b_G\alpha+(b_G+1)(n-\alpha-\beta)+(b_G+2)\beta\leq 2m-L_G,$$ 
which implies that
$$\alpha \geq L_G+ nb_G+n -2m + \beta \geq b_G(n+1) + n -2m,$$ 
since $L_G \geq b_G$ and $\beta \geq 0$.
Clearly $u_{b_G}^G \geq \alpha $, since deleting any $b_G$ edges incident with a vertex
of incidence degree $b_G$ results in an undecodable graph.
Since $G$ is decodable, no vertex of incidence degree $b_G$ is incident with all loops of $G$. 
Therefore, if $\alpha=b_G(n+1) + n -2m$, then $L_G=b_G$ and so 
$u_{b_G}^G \ge \alpha +1$ 
since we also have to count the case of deleting all $L_G$ loops from $G$. If $\alpha > b_G(n+1) + n -2m$, the result follows trivially.
\end{proof}% \mbox{}\\[-20pt]

%The lemma above clearly shows that  $c^G_{b_G}\leq \binom{m}{b_G} - b_G(n+1)-n+2m-1$.
We can also show a tight lower bound of $u^G_{b_G+y}$ when $y$ is small. 
We first note the following lemma which can be used in the latter lemma.
%We omit the proof due to the limited space. 
%\\[-18pt]

\begin{lemma}\label{lemub1eq}
Let $G\in \cD(n,m)$ satisfy
$u^G_{b_G} = b_G(n+1)+n-2m+1$. Then, with the same notation as in Lemma \ref{lemub1}, $\beta = 0$ and either
\begin{enumerate}
   \item
       $\alpha = b_G(n+1)+n-2m$ and $L_G=b_G$, or
   \item
       $\alpha =b_G(n+1) + n - 2m +1$ and $L_G = b_G+1$.
\end{enumerate} 
%Furthermore, if $m>2$, then we have  $\Delta_\ell(G) \leq b_G-1$ for each case.    
\end{lemma}
%\mbox{}\\[-4ex]
\begin{proof} 
Let $\theta = b_G(n+1)+n-2m$. Recall that, as in the proof of Lemma \ref{lemub1},
\begin{equation}\label{eqalpha}
\alpha \geq L_G+ nb_G+n -2m + \beta \geq b_G(n+1) + n -2m =\theta.
\end{equation}
Therefore, 
$$ \theta + 1 = u^G_{b_G} \geq \alpha \geq \theta,$$
so either $\alpha = \theta$, in which case $u^G_{b_G}= \alpha+1$,
or $\alpha = \theta + 1$ and $u^G_{b_G}=\alpha$.  

For the case $\alpha = \theta$, from (\ref{eqalpha}), we must have $L_G=b_G$ and $\beta = 0$. 
For the case  $\alpha = \theta + 1$, we have
$$\alpha = \theta+1 \geq L_G+ nb_G+n -2m + \beta=\theta - b_G+L_G + \beta,$$
which gives 
$b_G+1 \geq L_G + \beta \geq b_G+ \beta$.
Therefore, either $\beta = 1$ and $L_G=b_G$ or $\beta = 0$ and $L_G=b_G+1$.
Since for $\alpha=\theta+1$ we have $u^G_{b_G}=\alpha$, every undecodable subgraph of $G$ found by deleting $b_G$ edges is constructed by deleting the $b_G$ edges that meet a vertex of incidence degree $b_G$.
If $L_G=b_G$, then $G$ has a vertex of incidence degree $b_G$ that is incident with every loop of $G$, contradicting the decodability of $G$. We deduce that $\beta = 0$ and $L_G=b_G+1$.
\end{proof}%\mbox{}\\[-20pt]

Using Lemma~\ref{lemub1eq}, we obtain the following 
lemma giving a lower bound on $u^G_{b_G+y}$ for small values of 
$y$. Recall from Section~\ref{background} that 
$\Omega(G)$ denotes the maximum multiplicity of any edge joining a pair of vertices of $G$
and $\Delta_L(G):=\max\{ d_L(v) : v \in V\}$, where $d_L(v)$ is the number of loops incident with vertex $v$. 

\begin{lemma}\label{lemub2}
Let $G$ be a graph satisfying the hypothesis of Lemma \ref{lemub1eq}.
Let $\theta = b_G(n+1)+n-2m$.
Then
$$u_{b_G+y}^G\geq (\theta + 1)\binom{m-b_G}{y} + (n-\theta)\binom{m-b_G-1}{y-1},$$
for any $y$ satisfying $1 \leq y \leq b_G-\max(\Omega(G), \Delta_L(G))-1$. 
%where $\Omega_G$ is the maximum multiplicity of edges in $G$. 
\end{lemma}
%\mbox{}\\[-4ex]
\begin{proof} 
Let $y \in \{1,2, \ldots, b_G-\mu\}$, where $\mu=\max(\Omega(G)+1, \Delta_L(G)+1)$. Consider the following operations, each of which results in an undecodable subgraph of $G$ with $m-b_G-y$ edges.
\begin{enumerate}
   \item
      Delete $b_G$ edges incident with a vertex of incidence degree $b_G$ and delete a further 
      $y$ edges arbitrarily.      
   \item
      Delete $b_G+1$ edges incident with a vertex of incidence degree $b_G+1$ and delete a further $y-1$ edges arbitrarily. 
   \item
      Delete all $L_G$ loops of $G$, and then delete a further $b_G+y-L_G$ edges arbitrarily.   
\end{enumerate}
Observe first there are exactly $\alpha\binom{m-b_G}{y}$ (respectively $(n-\alpha)\binom{m-b_G-1}{y-1}$) ways to produce an undecodable subgraph by the operation 1) (respectively, by the operation 2)). 

The operations 1) and 2) are mutually exclusive, since in 1) at most $y \leq b_G-1$ edges are deleted from a vertex of incidence degree $b_G+1$. Moreover, the operations
2) and 3) are exclusive to each other, since 
when $b_G-y$ edges are deleted so that a vertex $v$ of incidence degree $b_G+1$ is isolated,
at most $d_L(v) + y - 1$ loops can be deleted and
\begin{eqnarray*}
d_L(v) + y - 1 &\leq& d_L(v) + (b_G-\Delta_L(G)-1) - 1\\
&=&b_G - (\Delta_L(G)-d_L(v)) - 2 \\
&\leq& b_G- 2 < L_G.
\end{eqnarray*}
Similarly, 1) and 3) are exclusive, since when a vertex $v$ of incidence degree $b_G$ is isolated,
 %$m(G)<L_G$, and 
at most $d_L(v)-y$ loops can be deleted and $d_L(v)-y\le b_G-1 <L_G$. 

%Now suppose that $\alpha=\theta$ and let $v\in V(G)$ such that $\delta_I(v)=m(G)$ and $\delta_\ell(v) \geq 1$. The following actions result in an undecodable subgraph by deleting some $m(G)$ edges of $G$.   
%\begin{enumerate}
 %  \item[(a)]
 %     Delete $m(G)$ edges incident with a vertex of incidence degree $m(G)$.
 %  \item[(b)]
 %     Delete all $m(G)$ loops.
 %  \item[(c)]
 %     Delete the $m(G)-\delta_\ell(v)$ non-loops edges incident with $v$ and delete the remaining $L_G - \delta_\ell(v)$ 
 %     loops of $G$ that are not incident with $v$.   
%\end{enumerate}
%Clearly under the assumption $\delta_\ell(v) \geq 1$, the operations (a),(b) and (c) are pairwise exclusive and so $\alpha+1 = u^G_{m(G)} \geq \alpha + 2$, giving a contradiction, so we deduce that no vertex of incidence degree $m(G)$ is incident with a loop. Then in 1), for a given vertex $v$ satisfying $\delta_I(v) = m(G)$, at most $\delta_\ell(v) + x =x \leq m(G)-1$ loops are deleted, which means 1) and 3) are mutually exclusive.

It follows that 
\begin{eqnarray*}
u_{b_G+y}^G & \geq & \alpha \binom{m-b_G}{y}+(n-\alpha)\binom{m-b_G-1}{y-1} \\
             &  +   & \binom{m-L_G}{b_G+y-L_G}, 
\end{eqnarray*}
 which yields for any $G \in \cD(n,m)$,
$$u_{b_G+y}^G\geq (\theta + 1)\binom{m-b_G}{y} + (n-\theta)\binom{m-b_G-1}{y-1}.$$
\end{proof}%\mbox{}\\[-20pt]

For given $u^G_{x}$, we can compute 
a lower bound on $u^G_{x+z}$ for $z\geq 0$ by using the following easy result.%\\[-18pt] 

\begin{lemma}\label{ub_lem}
   Let $G$ be a graph with $n$ vertices and $m$ edges.
   Then
   $$u^G_{x+z} \geq u^G_{x} \binom{m-x}{z} / \binom{x+z}{z}$$ 
   for any $z \geq 0$. 
\end{lemma}

\begin{proof}
From a undecodable graph of $m-x$ edges, 
we can generate $\binom{m-x}{z}$ undecodable subgraphs of $m-x-z$ edges.
On the other hand, given a undecodable graph $K$ of $m-x-z$ edges,
there are at most $\binom{x+z}{z}$ undecodable graphs of $m-x$ edges 
which have $K$ as a subgraph. 
Hence, there are at least  $u^G_{x} \binom{m-x}{z} / \binom{x+z}{z}$ 
undecodable graphs of $m-x-z$ edges. 
\end{proof}

The following corollary is now immediate.%\\[-18pt]

\begin{corollary}
Let $G \in \cD(n,m)$ satisfy the hypothesis of Lemma \ref{lemub1eq}. Let $\mu=\max(\Omega(G)+1, \Delta_\ell(G)+1)$
Then for each $z \geq 0$, we have 
\begin{eqnarray*}
u^G_{2b_G-\mu+z}  \geq 
u^G_{2b_G-\mu} \frac{\binom{m-2b_G + \mu}{z}} { \binom{2b_G-\mu+z}{z}}.
\end{eqnarray*} 
\end{corollary}
%\mbox{}\\

\subsection{For $GF(q)$ with $q$ odd} 
Finding good lower bounds on $u_{x}^G$ can be a hard task since 
not only loops but also odd cycles must be taken into consideration 
for determining whether a given graph is decodable or not. We have, up to this moment, 
the following sharp lower bound on $u^G_{2r}$ for some special class of graphs.

\begin{lemma}
Let $G$ be a graph in $\cD(n,rn)$, 
where $r\geq 2$ and $n\geq4$, with the minimum incidence 
degree $\delta_I(G)=2r$. Then we have
$$u_{2r}^G\geq n.$$
%$$c_{2r}^G\leq \binom{rn}{2r}-n.$$
\label{upper_lem}
\end{lemma}
%\mbox{}\\[-2ex]
\begin{proof}
First observe that if $\delta_I(G)=2r$, 
then each vertex has incidence degree $2r$ 
(\emph{i.e.} $G$ is a $2r$-regular graph). Indeed, 
if there exists a vertex with incidence degree 
strictly greater than $2r$, then 
we have $n\delta_I(G)<S_I(G)$, which is a contradiction since 
$2rn=n\delta_I(G)\le S_I(G) \le 2|E(G)|=2rn$. 

If $b_G=\delta_I(G)=2r$, then 
it is straightforward to check that 
$u_{2r}^G\geq n$ since for each vertex $v$ in $G$, deleting $2r$ edges attached to $v$ in $G$ generates an undecodable graph. Furthermore, if $b_G<\delta_I(G)=2r$, then we can construct
at least $\binom{rn-b_G}{2r-b_G}$ 
undecodable graphs with $rn-2r$ edges
by first generating an undecodable graph $K$ consisting of 
$rn-b_G$ edges, 
and then deleting another $2r-b_G$ edges arbitrarily from the 
edges in $K$. 
Since 
$$\binom{rn-b_G}{2r-b_G}\geq rn-b_G>rn-2r\geq n$$ holds 
by assumption on $n$ and $r$,
we have $u_{2r}^G\geq n$.
\end{proof}%\mbox{}\\[-20pt]

%%%%%%%%%
\begin{comment}
First fix $r=2$ in the equation above. Then, we can easily get that 
$2n-4\ge n$ holds whenever $n\ge 4$...(1)

Now let $n$ satisfy $rn-2r\ge n$ for some $r$. 
Then for $n$, we also have that $(r+1)n-2(r+1)\ge n$ holds...(2). 

Thus, from (1) and (2), for any $r\ge 2$ and $n\ge 4$, we have that
 $rn-2r\ge n$. 
\end{comment}
%%%%%%%%%%

\section{Encoding Schemes\label{candidate}}
%We have discussed so far the upper bounds on $b_G$ and $c_x^G$, but there is an another story whether there exists a graph satisfying those upper bounds. Indeed, if such examples exist, then we can suggest them as feasible candidates for good encoding schemes.  
In this section, we
introduce two algorithms that produce graphs meeting the bounds derived in the previous sections. 

Algorithm~\ref{interencoding_algo} yields an optimal coding scheme for $GF(q)$ with $q$ even. 
Note that the subscripts $i$ of the packets $p_i$ are computed modulo $n$ in what follows,
unless explicitly stated otherwise.

\begin{algorithm}[htb]
\caption {: A coding scheme for a system with $n$ packets
$p_{1}, p_2, \ldots, p_{n}$ and redundancy $r$.}
\begin{algorithmic}[1]
\REQUIRE {Let $k$ be a proper divisor of $n$; so $n=sk$ for some integer $s> 1$.}
%and $L_G=xk+y$ with $0\le y \le k-1$.}
%\STATE Consider a $k\times rs$ list such that its  $(x,y)$ entry ($1\le x \le k, 1\le y \le rs$) is $f_{(x-1)rs+y}$. 
\STATE Prepare a $k\times rs$ list such that  
its $(a,b)$ entry is $f_{b+(a-1)rs}$ with the column-wise order $h_c(a,b)=a+(b-1)k$.
\FORALL  {$(a,b)$ with $1\le a \le k, 1\le b \le rs$}
\IF {$1\le h_c(a,b)\le L_G $ }
	\STATE set $f_{b+(a-1)rs}:=p_{b+(a-1)s}$	
	\ELSIF {$b\le s-1$}
		\STATE set $f_{b+(a-1)rs}:=p_{b+(a-1)s}+ p_{b+as}$.
	\ELSIF {$s \le b\le rs-1$}
		\STATE set $f_{b+(a-1)rs}:=p_{b+(a-1)s}+ p_{b+(a-1)s+1}$.
	\ELSIF {$b= rs$}	
		\STATE set $f_{b+(a-1)rs}:=p_{b+(a-1)s}+ p_{1+(a-1)s}$.
	\ENDIF	
\ENDFOR
\RETURN {$\cC=\{f_{1}, f_{2}, \ldots, f_{rn}\}$ as a coding scheme.}
\end{algorithmic}
\label{interencoding_algo}
\end{algorithm}

\begin{example}
Consider the case of 9 packets, redundancy 2.
Suppose that $k=3$ and $L_G=4$ in Algorithm~\ref{interencoding_algo}. 

First prepare a $3\times 6$ list (together with the column-wise orders in brackets)  as shown below. 
\begin{table}[htb]
\begin{center}
\scriptsize
\begin{tabular}{|c|c|c|c|c|c|}
	\hline
	$f_1\  (1)$  &$f_2\ (4)$ & $f_3\ (7)$ &  $f_4\ (10)$ &  $f_5\ (13)$ &  $f_6\ (16)$ \\
	\hline
	$f_7\ (2)$  & $f_8\ (5)$ & $f_9\ (8)$ &  $f_{10}\ (11)$ &  $f_{11}\ (14)$ &  $f_{12}\ (17)$\\
	\hline
	$f_{13}\ (3)$  & $f_{14}\ (6)$ & $f_{15}\ (9)$ &  $f_{16}\ (12)$ & $f_{17}\ (15)$ & $f_{18}\ (18)$\\
	\hline
  \end{tabular}
\end{center}
\end{table}
\mbox{}\\
{\bf Step 1:} Since $L_G=4$, for all $(a,b)$ satisfying $1\le h_c(a,b) \le 4$ which are $(1,1), (2,1), (3,1), (1,2)$, 
set $f_1=p_1, f_7=p_4, f_{13}=p_7, f_2=p_2$.\\
{\bf Step 2:} Else,
\begin{itemize}
\item for $(a,b)$ with $b\le 2$, set $f_{b+(a-1)rs}=p_{b+(a-1)s}+ p_{b+as}$. 
\item for $(a,b)$ with $3\le b\le 5$, set $f_{b+(a-1)rs}=p_{b+(a-1)s}+ p_{b+(a-1)s+1}$.
\item for $(a,b)$ with $b=6$, set $f_{b+(a-1)rs}:=p_{b+(a-1)s}+ p_{1+(a-1)s}$.
\end{itemize}

Table~\ref{algo1ex_table} shows the resulting coding scheme $\cC$ under Algorithm~\ref{interencoding_algo}. 
\end{example}

\begin{table}[htb]
\begin{center}
\scriptsize
\begin{tabular}{|c|c|c|c|c|c|}
\hline
 \multicolumn{6}{|c|}{Encodings at the relays} \\
\hline
$p_1$  & $p_2$ & $p_3 + p_4$ &  $p_4 + p_5$ &  $p_5 + p_6$ &  $p_6 +p_1$ \\
\hline
$p_4$  & $p_5 + p_8 $ & $p_6+p_7$ &  $p_7 + p_8$ &  $p_8+ p_9$ &  $p_9+ p_4$\\
\hline
$p_7$  & $p_8+ p_{2} $ & $p_9+ p_{1}$ &  $p_{1}+ p_{2}$ & $p_{2}+ p_{3}$ & $p_{3}+ p_7$\\
\hline
\end{tabular}
\end{center}
\caption{The coding scheme under Algorithm~\ref{interencoding_algo} for 9 packets and redundancy 2 when 
$k=3$ and $L_G=4$ \label{algo1ex_table}}
\end{table}

Algorithm~\ref{interencoding_algo_2} produces optimal schemes over $GF(q)$ for $q$ odd.
Table~\ref{algo2ex_table} shows the coding scheme under Algorithm~\ref{interencoding_algo_2} for 9 packets and redundancy 2.

%\mbox{}\\[-18pt]

\begin{algorithm}[htbp]
\caption {: A coding scheme for a system with $n$ packets $p_{1}, p_2, \ldots, p_{n}$ and redundancy $r$.}
\begin{algorithmic}
\FORALL {$1\le i\le nr $ }
\STATE{set $f_{i}:=p_{i}+ p_{i+\lceil \frac{i}{n} \rceil}$.}
\ENDFOR
\RETURN {$\cC=\{f_{1}, f_{2}, \ldots, f_{rn}\}$ as a coding scheme.}
\end{algorithmic}
\label{interencoding_algo_2}
\end{algorithm}

\begin{table}[htb]
\begin{center}
\scriptsize
\begin{tabular}{|c|c|c|c|c|c|}
\hline
 \multicolumn{6}{|c|}{Encodings at the relays} \\
\hline
$p_1 + p_2$  & $p_2 + p_3$ & $p_3 + p_4$ &  $p_4 + p_5$ &  $p_5 + p_6$ &  $p_6 + p_7$ \\
\hline
$p_7 + p_8$  & $p_8 + p_9$ & $p_9 + p_1$ &  $p_1 + p_3$ &  $p_2 + p_4$ &  $p_3 + p_5$ \\
\hline
$p_4 + p_6$  & $p_5+  p_7$ & $p_6 + p_8$ &  $p_7 + p_9$ &  $p_8 + p_1$ &  $p_9 + p_2$ \\
\hline
\end{tabular}
\end{center}
\caption{The coding scheme under Algorithm~\ref{interencoding_algo_2} for 9 packets and redundancy 2 \label{algo2ex_table}}
\end{table}

Recall from Subsection~\ref{b_G-even} that 
$b_G\le 2r-1$ for $GF(q)$ with $q$ even.
The following proposition 
shows that it is indeed possible to generate 
a graph $G$ for which $b_G= 2r-1$ holds based on Algorithm~\ref{interencoding_algo}.
%Note that the subscripts $i$ of packets (vectors) $p_i$ are computed modulo $n$ in what follows,
%if not indicated explicitly. 

\begin{proposition}\label{propalg1}
Let $\cC$ be the coding scheme for a system sending $n$ packets defined as in Algorithm~\ref{interencoding_algo}, where the integers
$k$ and $r$ satisfy $k,r \geq 2$.
 Let $s=n/k$ and let the graph representation $G$ of $\cC$ satisfy
$L_G\geq 2r-1$. If $k\leq L_G\leq (s-1)k$ then it holds that
$b_G=\delta _{I}(G)=2r-1$.
\label{match_prop}
\end{proposition}
%\mbox{}\\[-4ex]
\begin{proof} 
Let $t=sr$. 
Observe that the graph representation $G$ of $\cC$ under Algorithm~\ref{interencoding_algo} satisfies 
the following properties.
\begin{enumerate}
\item Any $(a,1)$ with $1\le a \le k$ satisfies $h_c(a,1)=a\le k\le L_G$.  
Since $n=sk$ by assumption, 
each vertex $p_i$ with $i \equiv 1 \pmod s$ has a loop.
\item Any $(a,b)$ such that $b\ge s$ satisfies $h_c(a,b)\ge a+(s-1)k > L_G$, 
and therefore, 
 lines 7-8 in the algorithm are applied to all pairs $(a,b)$ with $(s-1)k \le b \le rs-1$. 
Since there are in total $sk(r-1)=n(r-1)$ of such pairs,  
for each $1\leq i \leq n$, the number of edges between vertices $p_i$ and $p_{i+1}$ is $r-1$.
\item From 2), $G$ has a connected subgraph consisting of multi-edges $(p_i, p_{i+1})$, $1\le i \le n$, and  hence, $G$ itself is connected.  
Furthermore, $G$ cannot be disconnected without deleting the multi-edges $(p_i,p_{i+1})$ and $(p_{j},p_{j+1})$ for some pair $(i,j)$ with $i\not =j$.
\item If a vertex $p_i$ does not have a loop, then
\begin{itemize}
\item it is adjacent to vertex $p_{i+s}$ when $i\not\equiv 0 \pmod s$
from lines 5-6 in the algorithm. 
\item it is adjacent to vertex $p_{i-t+1}$ otherwise
from lines 9-10 in the algorithm. 
\end{itemize}
\end{enumerate}

From the statements above,
it is straightforward to see that $G$ is connected and 
$\delta _{I}(G)=2r-1$ holds, 
which automatically implies $b_G\leq 2r-1$ from 1) in Remark~\ref{minloop_rem}. 
So we focus in the rest of the proof that 
$b_G\geq 2r-1$ also holds. 

Now suppose that for some pair $(i,j)$ with $i\neq j$, the multi-edges $(p_i,p_{i+1})$ and $(p_{j},p_{j+1})$ are deleted from $G$ (so $2(r-1)$ edges are deleted in total),  
and call the resulting graph $\widehat G$.
Denote by $H_i$ the subgraph of $\widehat G$ induced by the vertices $p_{i+1},p_{i+2},\ldots,p_{j}$, %$\pmod n$, 
and by $H_{j}$ the one induced by the vertices $p_{j+1},p_{j+2},\ldots,p_{i}$. % $\pmod n$. 
%Then, if we see $H_a$ and $H_{a'}$ within $G$,  either there exists some edges 
%between $H_a$ and $H_{a'}$ or not. 
If both $H_i$ and $H_{j}$ contain loops, then we can conclude that 
$b_G\geq 2r-1$. 
Furthermore, it cannot happen that neither $H_i$ nor $H_{j}$ have loops since 
$L_G=L_{\widehat G}\geq 1$. Therefore, we need only to consider the case
for which $H_i$ contains a loop but $H_{j}$ does not. In this case, 
we will focus on $H_i$ and $H_{j}$ as subgraphs of $\widehat G$, and show the existence of an edge in $E(\widehat G)$ joining them, which implies that $b_G\geq 2r-1$. 

As $H_{j}$ does not contain loops, $|V(H_{j})|<s$ since otherwise,  
at least one of the vertices $p_\ell$ in $H_{j}$ satisfies 
$\ell \equiv 1 \pmod s$, 
and therefore, $H_{j}$ contains a loop from Property 1).

If $H_{j}$ contains a vertex $p_\ell$ with $\ell\equiv 0 \pmod s$, 
$p_\ell$ is adjacent to the vertex $p_{\ell-t+1}$, where $\ell-t+1 \equiv 1 \pmod s$ as $t=sr$. 
Since $p_{\ell-t+1}$ has a loop from Property 1), it is in $H_i$. 
If each vertex $p_\ell$ in $H_{j}$ satisfies $\ell\not\equiv 0 \pmod s$, then $p_\ell$ is adjacent to 
$p_{\ell+s}$. Since $|V(H_{j})|<s$ and $|V(H_{i})|>n-s=(k-1)s\geq s$, $p_{\ell+s}$ is in $H_i$.
In each case, there exists an edge in $E(\widehat G)$ joining $H_i$ and $H_{j}$ as required. 
\end{proof}
%\mbox{}\\[-20pt]

Furthermore, we can also prove that when $2r-1\leq L_G \leq 2r$ the graph $G$ in Proposition~\ref{match_prop} satisfies 
$u^G_{b_G}=u^G_{2r-1}= 2r$, which is the 
lower bound of $u^G_{b_G}$ obtained in Lemma~\ref{lemub1}, as stated below.
%The proof will be provided in the full version of this paper.

\begin{proposition}
Let $G$ be the graph satisfying the conditions described in Proposition~\ref{match_prop}.
If $k\geq 3, r\geq 2$ and $L_G$ is either $2r-1$ or $2r$, then we have 
$u^G_{b_G}=u^G_{2r-1}= 2r$.
\label{ub_cor}
\end{proposition} 
%\mbox{}\\[-4ex]
\begin{proof} 
%Now suppose $b_G=2r-1$.
It is straightforward to check from the construction of $G$ 
that the number of vertices with incidence degree $2r-1$ (\emph{resp.} $2r$) 
is $L_G$ (\emph{resp.} $n-L_G$). 
So when $L_G=2r-1$ or $L_G=2r$, $2r$ undecodable subgraphs 
can be found by deleting $2r-1$ edges from $G$ by 
\begin{enumerate}
	\item[(1)] making a vertex $v$ with $d_I(v)=2r-1$ isolated;
	\item[(2)] deleting all $L_G$ loops from $G$ (when $L_G=2r-1$)
\end{enumerate}
We will show that we cannot construct 
other undecodable subgraphs by deleting $2r-1$ edges from $G$ (that is, 
other undecodable subgraphs of $G$ with exactly $nr-2r+1$ edges).

Assume by contradiction that
there exists an undecodable subgraph $H$ of $G$ with $nr-2r+1$ edges constructed by 
neither (1) nor (2) above.
Then $H$ must be disconnected since otherwise, we have to delete all $L_G\geq 2r-1$ loops to make $H$ undecodable. Therefore, from 3) in the proof of Proposition~\ref{propalg1}, 
$2r-2$ edges $(p_i,p_{i+1})$ and  $(p_j,p_{j+1})$, $i\not =j$, must be deleted from $G$ 
to generate $H$.   

As before, let  $\widehat G$ be the subgraph of $G$ found by deleting the multi-edges $(p_i,p_{i+1})$ and $(p_{j},p_{j+1})$ for some $i\neq j$. 
Again, let $H_i$ and $H_{j}$ be the subgraphs of $\widehat G$
induced by the vertices $p_{i+1},p_{i+2},\ldots,p_{j}$,  
and $p_{j+1},p_{j+2},\ldots,p_{i}$, respectively. 

First suppose $r\ge 3$, so that neither $H_i$ nor $H_{j}$ can be disconnected 
by deleting another a single edge.
If $|V(H_i)|=1$, say $V(H_i)=\{ v \}$ for some $v \in V$, 
then $d_I(v)=2r$ in $G$ since the case of $d_I(v)=2r-1$ has been counted in (1). 
As $d_I(v)=2r$ and $2r-2$ edges amongst these $2r$ edges have been deleted already, 
there are two remaining edges attached to $v$. 
Since a non-loop edge attached to $v$ (if it exists) 
is joined with a vertex in $H_{j}$, we can 
conclude that we need to delete two or more edges from $\widehat G$ to make $H$, 
which shows that $|E(H)|<nr-2r+1$. 

So now assume that $|V(H_i)|, |V(H_{j})|\geq 2$. 
If both $H_i$ and $H_{j}$ contain at least two loops, 
it is trivial that at least two edges must be deleted to make $H$ from $\widehat G$. 
If $H_j$ contains no loops, then we can use the same argument in the proof of 
Proposition~\ref{match_prop} to show the existence of at least $|V(H_j)|\geq 2$ edges 
joining $|V(H_i)|$ and $|V(H_{j})|$. 
Furthermore, if $H_j$ contains only one loop 
(in which case $H_{i}$ has at least two loops since $L_G\geq 2r-1\geq 3$), 
then we have  $2\leq|V(H_j)|\leq 2s-1$, since $H_j$ contains 
at least two loops whenever $|V(H_j)|>2s-1$. 
Thus, $|V(H_{i})|\geq n-(2s-1)=s(k-2)+1>s$ as $k\geq 3$ by assumption.
Hence, from 4) in the proof of Proposition~\ref{propalg1} we have there exists 
an edge joining $H_i$ and $H_{j}$.   
For each case, we can conclude that $|E(H)|<nr-2r+1$, and 
hence, it is impossible to make an undecodable graph with $nr-2r+1$ edges 
except for deleting edges according to (1) and (2), as required. 

We next suppose that $r=2$. If neither $H_i$ nor $H_{j}$ can be disconnected 
by deleting another a single edge, then use the same argument above. 
So suppose that $H_j$ can be disconnected by deleting a single edge.
 
As before we can assume that $|V(H_i)|, |V(H_{j})|\geq 2$. 
Now, assume that an edge $(a,a+1)$ is deleted from $H_j$ and 
the resulting graph is disconnected. There are two cases to consider.
In the first case we suppose that $H_i$ contains a loop.
If  $|V(H_j)|<s$, each vertex in $H_j$ 
with no loops is adjacent to some vertex in $H_{j}$. 
If $|V(H_j)|\geq s$, at least one vertex in $H_j$ has a loop, 
which implies that at least one of two components in $H_j$ has a loop. 
If each of the connected components of $H_j$ contains loop, then we are done, as $H_j$ is decodable.
If there exists a component with no loops, then it has at most $s-1$ vertices and  
some of these are adjacent to vertices of $H_{j}$ or the other component with loops. 
For each case, we can conclude that the resulting graph is decodable. 

Now suppose that $H_{i}$ has no loops (so all loops are in $H_j$). Then $|V(H_{i})|<s$. 
Also each vertex in $H_{i}$ is adjacent to 
some vertex in $H_j$. If both of two components in $H_j$ contain loops, or if $H_{i}$ contains two vertices such that one vertex is joining to 
one component in $H_j$ and the other vertex is joining to the other component in $H_j$, then we are done.
If each vertex in $H_{i}$ is adjacent to vertices in a component in $H_j$ with no loops, then the size of the component is also bounded by $s-1$. 
Thus, $n$ vertices are partitioned into 3 parts in the resulting graph, each of which consists of consecutive integers, so that
2 parts contain $s-1$ or fewer vertices all of which are not equivalent to 1 modulo $s$, yielding a contradiction.
\end{proof}
%\mbox{}\\[-18pt]

We have the following proposition, for Algorithm~\ref{interencoding_algo_2}.　 
%describing its efficiency for odd values of $q$. 

\begin{proposition}
Let $\cC$ be the coding scheme for a system sending $n$ packets and redundancy $r$, 
where $n>3$, 
defined as in Algorithm~\ref{interencoding_algo_2}. Suppose that $n/2<4r<n$. 
Then the representation $G=G_{\cC}$ of 
$\cC$ satisfies
$b_G=\delta _{I}(G)=2r$, and 
furthermore, that 
$u^G_{b_G}=u^G_{2r}= n$.
\label{algo2_prop}
\end{proposition}
%\mbox{}\\[-4ex]
\begin{proof} 
Observe that each vertex $p_i$ of $G$ is adjacent to $2r$ vertices
$p_{i\pm k}$, $1\leq k \leq r$, and hence, 
$d_I(p_i)=2r$ for each $i$.
We claim that except for the 
case of deleting $2r$ edges incident to a vertex, deleting $2r$ or fewer edges results in a decodable subgraph.

We first show that deleting $2r$ or fewer edges rather than $2r$ incident edges to a vertex 
always gives us a connected graph. Consider a cut of vertex set $V$ of $G$; so suppose that $V$ is 
partitioned into two sets $V_1$ and $V_2$. Let $|V_1|=t>1$ and assume 
without loss of generality that $t \le n/2 $ since 
the case of $t > n/2$ implies that $|V_2|=n-t \le n/2$ %(and hence, $|V_2|\le n/2$) 
 and the 
argument below follows by replacing $V_1$ with $V_2$. 

Since there are no multi-edges in $G$,
at most $\binom{t}{2}$ edges are used to connect vertices in $V_1$, and 
therefore, the size of the cut-set is at least $2rt-\binom{t}{2}$.  
From the following inequality 
\begin{eqnarray*}
& &2rt-\binom{t}{2}>2r\\
%&\Leftrightarrow& 4rt-t(t-1)-4r>0\\
&\Leftrightarrow&(t-1)(4r-t)>0 
\end{eqnarray*}
we have that the size of the cut-set is bigger than $2r$
as long as $1<t<4r$. 
Therefore, we can conclude that the resulting graph is connected
whenever $n/2<4r$.

We next show that any resulting graph after deletion of edges has an odd cycle. 
Consider a subgraph $\tilde G$ of $G$ consisting of $n$ triangles 
(cycles of length 3) $\tau_i: p_i, p_{i+1}, p_{i-1}, p_i$ with $1\le i\le n$. 
It is trivial that $\tau_i$ and $\tau_{i'}$  are edge-disjoint (\emph{i.e.}, they do not share the same edges) if $i$ and $i'$ are both even, and hence, there exist  $\lfloor n/2 \rfloor$  edge-disjoint triangles in $\tilde G$. 
Thus, at least $\lfloor n/2 \rfloor$ edges should be deleted from $\tilde G$ to make a 
decodable graph.

If $n$ is even, then  $n/2 = \lfloor n/2 \rfloor = \lceil n/2 \rceil$. 
If $n$ is odd, then $\lfloor n/2 \rfloor = (n-1)/2$. 
However, since each edge is used within at most 2 triangles,
at least another edge must be deleted to remove all $n$ triangles, which results in a total of $(n-1)/2+1 = (n+1)/2 = \lceil n/2 \rceil$ edges. 

It follows that for any $n$, at least $\lceil n/2 \rceil$ edges must be deleted if the resulting graph is undecodable.
Then as long as $n/2 > 2r$, the result holds.
 \end{proof}
%\mbox{}\\[-18pt]

\begin{comment}
For graphs $\tilde G$ in $D(n,rn)$, 
the sum of the incidence degrees satisfies $S_I(\tilde G)\le 2|E(\tilde G)|=2rn$,
and hence, for its minimum incidence degree, we have  $\delta_I(G')\leq 2r$ from Lemma~\ref{edgenumber_lem}.
Thus, $G$ in Proposition~\ref{algo2_prop} is a graph which has the largest minimum incidence degree 
amongst all graphs in $D(n,rn)$. 
We conjecture that any graph obtained from 
Algorithm~\ref{interencoding_algo_2} 
by setting $r=2$ and $n\geq 9$ satisfies the condition. 
\end{comment}

%To summarize, propositions in this section can be used to consider a coding scheme with high decoding probability. 
%We hope to analyze properties of coding schemes in detail with a hope to provide a best coding scheme for a given system.  

%\begin{corollary}\label{ub_2r-1_lem}
%   Let $G \in {\cC}(n,rn)$ be a graph representation for a coding scheme %such that $m(G)=2r-1$.
%   Then $c^G_{2r-1} \leq \binom{rn}{2r-1}-2r$. Moreover, if equality holds then
%   $c^G_{2r+z} \leq \binom{rn}{2r+z}-(2rn^2-4r^2+n+1)\binom{rn-2r}{z} / %\binom{2r+z}{z}$
%   for each $z \geq 0.$
%\end{corollary}

\section{Comparison of Various Coding Schemes \label{comparison}}
In this section, we will analyze robustness against packet loss 
for the coding scheme given by Algorithm~\ref{interencoding_algo} and Algorithm~\ref{interencoding_algo_2}. 
We compare each coding scheme over $GF(q)$ with $q$ even and with $q$ odd by computing its 
decoding probability.

Throughout this section, set the number of packets $n=9$ and redundancy $r=2$. 
We consider the following 12 graphs  for comparison. 
%(thus, the number of packets received by the terminal when no packet loss 
\begin{itemize}
\item $G_0, G_1, \ldots, G_9$ are 
the graph representations of coding schemes obtained 
by Algorithm~\ref{interencoding_algo} with divisor $k=3$,
where we set the number of uncoded packets (\emph{i.e.} the number of loops) to be $0,1, \ldots, 9$, respectively.  
\item  $G'$ is the graph representation of a coding scheme 
obtained by Algorithm~\ref{interencoding_algo_2}.
\item $G$ be the graph corresponding to transmitting packets without coding (\emph{i.e.}, each packet is sent twice to the terminal without coding). 
\end{itemize}
%as we proved that the Algorithm produces a reasonably good coding scheme over $GF(q)$ with $q$ even. 

Table~\ref{bH_table} presents $b_H$ for each graph $H$ for $GF(q)$ with $q$ even and with $q$ odd.  
\begin{table}[htpb]
\begin{center}
%\begin{mimiskip}
\scriptsize
\begin{tabular}{|c||c|c|}               
\hline
 Graph  $H$     &   $b_H$ with $q$ even &  $b_H$ with  $q$ odd \\
\hline    
$G_0$ &  0 & 3  \\
\hline        
$G_1$ &  1 & 3  \\
\hline
$G_2$ &  2 & 3  \\
\hline
$G_3$ &  3 & 3  \\
\hline
$G_4$ &  3 & 3  \\
\hline
$G_5$ &  3 & 3  \\
\hline
$G_6$ &  3 & 3  \\
\hline
$G_7$ &  3 & 3  \\
\hline
$G_8$ &  3 & 3  \\
\hline
$G_9$ &  3 & 3  \\
\hline
$G'$ & 0 & 4 \\
\hline
$G$   &  2 & 2 \\
\hline
\end{tabular}
%\end{miniskip}
\end{center}
\caption{$b_H$ for each graph $H$}
\label{bH_table}
\end{table}
%\mbox{}\\[-20pt]

Tables~\ref{dp_table_0.6}, \ref{dp_table_0.7} and \ref{dp_table_0.8} provide the  
decoding probabilities $\cP_H(p,1-p)$ for each graph $H$ when $p=0.6$, $p=0.7$ and $p=0.8$, respectively. Recall that $p$ is the probability that each packet is successfully transmitted, 
and graphs with higher decoding probabilities are preferred. 

\begin{comment}
\begin{table}[!h]
\begin{center}
%\begin{mimiskip}
\scriptsize
\begin{tabular}{|c||c|c|}               
\hline
 Graph      &   Decoding Prob. with $q$ even &   Decoding Prob. with $q$ odd \\
\hline    
$G_0$ &  0        &  0.344643 \\
\hline        
$G_1$ &  0.206551 &  0.359737 \\
\hline
$G_2$ &  0.295723 &  0.366817  \\
\hline
$G_3$ &  \color{red}{0.332825} &  0.365883  \\
\hline
$G_4$ &  0.330986 &  0.342552  \\
\hline
$G_5$ &  0.316544 &  0.321766  \\
\hline
$G_6$ &  0.300926 &  0.303524 \\
\hline
$G_7$ &  0.279510 &  0.280647  \\
\hline
$G_8$ &  0.259285 &  0.259663 \\
\hline
$G_9$ &  0.240471 &  0.240475 \\
\hline
$G'$ & 0 &  \color{blue}{0.413086} \\
\hline
$G$   &  0.075085 &  0.075085 \\
\hline
\end{tabular}
%\end{miniskip}
\end{center}
\caption{The decoding probabilities when $p=0.5$ \label{dp_table_0.5}}
\end{table}
\end{comment}

\begin{table}[!h]
\begin{center}
%\begin{mimiskip}
\scriptsize
\begin{tabular}{|c||c|c|}               
\hline
 Graph      &   Decoding Prob. with $q$ even &   Decoding Prob. with $q$ odd \\
\hline    
$G_0$ &  0        &  0.623651 \\
\hline        
$G_1$ &  0.405309 &  0.639395 \\
\hline
$G_2$ &  0.554459 &  0.646345  \\
\hline
$G_3$ &  \color{red}{0.607826} &  0.644499  \\
\hline
$G_4$ &  0.605533 &  0.617078  \\
\hline
$G_5$ &  0.587424 &  0.592104  \\
\hline
$G_6$ &  0.567454 &  0.569576 \\
\hline
$G_7$ &  0.541135 &  0.542086  \\
\hline
$G_8$ &  0.515860 &  0.516201  \\
\hline
$G_9$ &  0.491854 &  0.491856 \\
\hline
$G'$ & 0 &  \color{blue}{0.703957} \\
\hline
$G$   &  0.208216 &  0.208216 \\
\hline
\end{tabular}
%\end{miniskip}
\end{center}
\caption{The decoding probabilities when $p=0.6$ \label{dp_table_0.6}}
\end{table}

\begin{table}[!h]
\begin{center}
%\begin{mimiskip}
\scriptsize
\begin{tabular}{|c||c|c|}               
\hline
 Graph      &   Decoding Prob. with $q$ even &   Decoding Prob. with $q$ odd \\
\hline    
$G_0$ &  0        &  0.847225 \\
\hline        
$G_1$ &  0.608349 &  0.856136 \\
\hline
$G_2$ &  0.784263 &  0.859695  \\
\hline
$G_3$ &  \color{red}{0.834280} &  0.857901  \\
\hline
$G_4$ &  0.832605 &  0.838849  \\
\hline
$G_5$ &  0.818904 &  0.821020  \\
\hline
$G_6$ &  0.803590 &  0.804415 \\
\hline
$G_7$ &  0.784373 &  0.784757  \\
\hline
$G_8$ &  0.765603 &  0.765757  \\
\hline
$G_9$ &  0.747395 &  0.747396 \\
\hline
$G'$ & 0 &  \color{blue}{0.902653} \\
\hline
$G$   &  0.427930 &  0.427930 \\
\hline
\end{tabular}
%\end{miniskip}
\end{center}
\caption{The decoding probabilities when $p=0.7$ \label{dp_table_0.7}}
\end{table}

\begin{table}[!h]
\begin{center}
%\begin{mimiskip}
\scriptsize
\begin{tabular}{|c||c|c|}               
\hline
 Graph      &   Decoding Prob. with $q$ even &   Decoding Prob. with $q$ odd \\
\hline    
$G_0$ &  0        &  0.961702 \\
\hline        
$G_1$ &  0.772801 &  0.964009 \\
\hline
$G_2$ &  0.925741 &  0.964775  \\
\hline
$G_3$ &  \color{red}{0.955810} &  0.963999  \\
\hline
$G_4$ &  0.955182 &   0.956759\\
\hline
$G_5$ &  0.949392 &  0.949776  \\
\hline
$G_6$ &  0.942937 &  0.943050 \\
\hline
$G_7$ &  0.935311 &  0.935368  \\
\hline
$G_8$ &  0.927755 &  0.927783  \\
\hline
$G_9$ &  0.920293 &  0.920293 \\
\hline
$G'$ & 0 &  \color{blue}{0.982510} \\
\hline
$G$   &  0.692534 &  0.692534 \\
\hline
\end{tabular}
%\end{miniskip}
\end{center}
\caption{The decoding probabilities when $p=0.8$ \label{dp_table_0.8}}
\end{table}

\begin{comment}
\begin{table}[!h]
\begin{center}
%\begin{mimiskip}
\scriptsize
\begin{tabular}{|c||c|c|}               
\hline
 Graph      &   Decoding Prob. with $q$ even &   Decoding Prob. with $q$ odd \\
\hline    
$G_0$ &  0        &  0.996226 \\
\hline        
$G_1$ &  0.896783 &  0.996366 \\
\hline
$G_2$ &  0.986356 &  0.996396  \\
\hline
$G_3$ &  \color{red}{0.995294} &  0.996317  \\
\hline
$G_4$ &  0.995228 &  0.995331  \\
\hline
$G_5$ &  0.994343 &  0.994355  \\
\hline
$G_6$ &  0.993387 &  0.993389 \\
\hline
$G_7$ &  0.992339 &  0.992340 \\
\hline
$G_8$ &  0.991293 &  0.991293  \\
\hline
$G_9$ &  0.990248 &  0.990248 \\
\hline
$G'$ & 0 &  \color{blue}{0.999044} \\
\hline
$G$   &  0.913517 &  0.913517 \\
\hline
\end{tabular}
%\end{miniskip}
\end{center}
\caption{The decoding probabilities when $p=0.9$ \label{dp_table_0.9}}
\end{table}
\end{comment}

%Table~\ref{dp_table} shows an improvement in the decoding probabilities 
%between $GF(q)$ with $q$ even and $GF(q)$ with $q$ odd. 
%Not surprisingly, the simulation results and theoretical results are very close.  
%Furthermore, those results support the accuracy of our propositions 
%provided in this paper and \cite{BMMP:11}. 
%$\cP_{G_3}$ is the highest amongst those 
%coding schemes, which supports the accuracy of results provided in this paper and \cite{BMMP:11}. Also, 
%It is interesting to note that 
Observe that the decoding probability of a coding scheme over $GF(q)$ with $q$ odd is higher
that over $GF(q)$ with $q$ even;  %since the $2$-rank of an incidence matrix is more frequently lower than
%its $p$-rank for odd $p$. In particular, 
for coding schemes over fields of even characteristic decoding relies 
solely on the presence of loops in each connected component,
whilst over fields of odd characteristic any odd cycle in each component will suffice. 
Furthermore, we can also confirm that the decoding probability of $H$ increases as $b_H$ gets larger.

More importantly, $G_3$ (a graph with $3$ loops) and $G'$ give the highest decoding probability amongst those 12 coding schemes over $GF(q)$ with $q$ even and  with $q$ odd, respectively,
which supports the results provided in this paper.
Figures~\ref{dpeven_fig} and \ref{dpodd_fig} show the  
decoding probabilities $\cP_H(p,1-p)$ for $G_3, G'$ and $G$ over $GF(q)$ with $q$ even with $q$ odd,  
respectively, for comparison. From these figures, we can see that decoding probability can be increased by up to 
0.42 (at $p=0.65$) for $q$ even and 0.51 (at $p=0.64$) for $q$ odd,
by using the proposed low-power coding schemes. 

\begin{figure}[htb]
\begin{center}
\includegraphics[width=85mm]{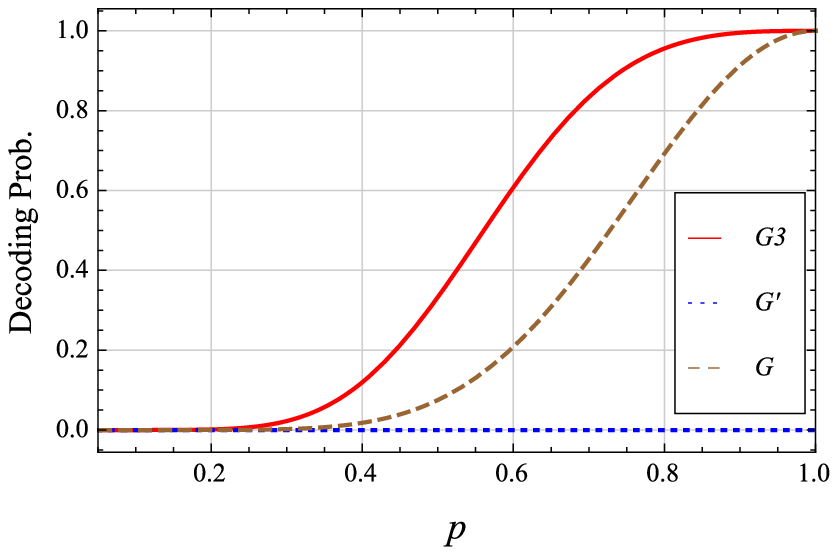}
\caption
{The decoding probabilities of $G_3, G'$ and $G$ over $GF(q)$ with $q$ even. \label{dpeven_fig}}
\end{center}
\end{figure}

\begin{figure}[htb]
\begin{center}
\includegraphics[width=85mm]{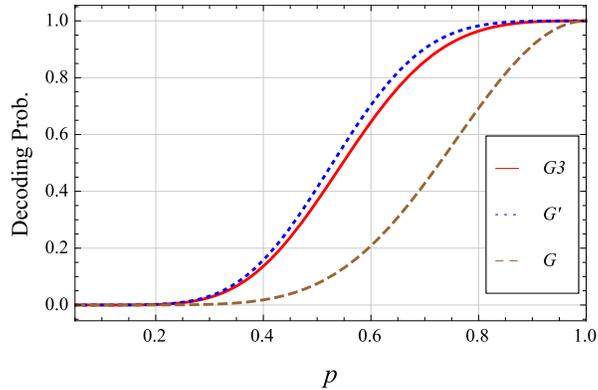}
\caption
{The decoding probabilities of $G_3, G'$ and $G$ over $GF(q)$ with $q$ odd. \label{dpodd_fig}}
\end{center}
\end{figure}

\section{Final Remarks \label{conclusion}}

In this paper, we theoretically analyzed the robustness of 
coding schemes against packet loss, using only additions of two packets over $GF(q)$.
Such coding schemes 
are well-suited to systems for which 
data reliability and low computational complexity are 
strongly preferred (\emph{e.g.} WBANs since 
they require little energy for encoding and decoding). 
We introduced some criteria for a coding scheme to have 
high decoding probability using graph theory.
We also compared decoding probabilities of different schemes. Our results suggest that
coding schemes defined over $GF(q)$ with $q$ odd may outperform their 
counterparts defined over $GF(q)$ with $q$ even.

The results here are related to problems such as the number of bipartite subgraphs of a graph and the size of a maximal cut of a graph. 
We remark that for linear functions $f_i(p_1,p_2, \ldots, p_n)=\sum_{j \in J}p_{i_j}$ with $J \subset \{1,2, \ldots, n\}$ of size greater than 2,
a hypergraph representation can apply. Then as in Lemma \ref{lemdh}, the scheme is decodable if and only if its incidence matrix has full rank $n$.
%As an example, consider the incidence matrix $M$ of a $2-(n,k,\lambda)$ design of order $r = \lambda (n-1)/(k-1)$. Then $M$ has full rank over $GF(q)$ if and only if $p$ does not divide $r(r-\lambda)$, where $q$ is a power of  a prime $p$ \cite{H:73}.

\bibliographystyle{ieeetr}
\bibliography{wban}

\end{document}